\documentclass[pra,twocolumn,superscriptaddress]{revtex4}
\usepackage{amsthm}
\usepackage{amsmath}
\usepackage{latexsym}
\usepackage{amsfonts}
\usepackage{amssymb}
\usepackage{color}
\usepackage{bbm,dsfont}
\usepackage{graphicx}
\usepackage{hyperref}
\usepackage{subfigure}
\usepackage{bm}

\usepackage[normalem]{ulem}
\usepackage{mathtools}
\newtheorem{proposition}{Proposition}
\newtheorem{proposition?}{Proposition?}
\newtheorem*{proposition*}{Proposition}

\newtheorem{lemma}{Lemma}

\theoremstyle{definition}

\newtheorem{definition}{Definition}
\newtheorem{observation}{Observation}
\newtheorem{observation*}{Observation}

\newcommand{\sh}{\mathcal{S(H)}} %states
\newcommand{\ket}[1]{|#1\rangle} %ket
\newcommand{\bra}[1]{\langle#1|} %bra
\newcommand{\kb}[2]{|#1\rangle\langle#2|} %ketbra
\newcommand{\braket}[2]{\langle #1 | #2 \rangle} % for Dirac brakets
\newcommand{\ketbra}[2]{\left| #1 \rangle\langle #2 \right|} % for

\newcommand{\no}[1]{\left\|#1\right\|} %norm
\newcommand{\tr}[1]{\textrm{tr}\left[#1\right]} %trace
\newcommand{\id}{\mathbbm{1}} %identity operator
\newcommand{\var}{\textrm{Var}} %variance
\newcommand{\ve}{\mathbf{e}} %e
\newcommand{\vsigma}{\boldsymbol{\sigma}} %sigma
\newcommand{\E}{\mathsf{E}}%generic observable
\newcommand{\F}{\mathsf{F}}%generic observable
\newcommand{\G}{\mathsf{G}}%generic joint observable
\renewcommand{\P}{\mathsf{P}}%sharp observable
\newcommand{\M}{\mathsf{M}}%observable
\newcommand{\N}{\mathsf{N}}%observable

\newcommand{\ze}{\mathbf{0}}

\newcommand\numberthis{\addtocounter{equation}{1}\tag{\theequation}}

\newcommand{\cft}{Center for Theoretical Physics, Polish Academy of Sciences, Al. Lotników 32/46, 02-668 Warszawa, Poland}
\newcommand{\nasa}{USRA Research Institute for Advanced Computer Science (RIACS), Mountain View, CA, 94043, USA}

\begin{document}
\title[]{Estimating Quantum Hamiltonians via Joint Measurements \\of Noisy Non-Commuting Observables}

\

\date{March 13, 2023}
\author{Daniel McNulty}
\affiliation{\cft}
\author{Filip B. Maciejewski}
\affiliation{\cft}
\affiliation{\nasa}
\author{Michał Oszmaniec}
\affiliation{\cft}

\begin{abstract}
Estimation of expectation values of incompatible observables is an essential practical task in quantum computing, especially for approximating energies of chemical and other many-body quantum systems.
In this work we introduce a method for this purpose based on performing a single joint measurement that can be implemented locally and whose marginals yield noisy (unsharp) versions of the target set of non-commuting Pauli observables. We  derive bounds on the number of experimental repetitions required to estimate energies up to a certain precision. We compare this strategy to the classical shadow formalism and show that our method yields the same performance as the locally biased classical shadow protocol. We also highlight some general connections between the two approaches by showing that classical shadows can be used to construct joint measurements and vice versa. Finally, we adapt the joint measurement strategy to minimise the sample complexity when the implementation of measurements is assumed noisy. This can provide significant efficiency improvements compared to known generalisations of classical shadows to noisy scenarios. 
\end{abstract}

\maketitle

\section{Introduction}

Measurement incompatibility, one of the defining non-classical features of quantum theory, limits an observer's ability to measure certain physical properties of a system simultaneously. While typically viewed as a quantum resource  \cite{heinosaari15,chitambar19}, essential in applications such as non-locality \cite{wolf09}, quantum steering \cite{quintino14,uola15}, and state discrimination \cite{skrzypczyk19,carmeli19}, incompatibility is a major issue for variational quantum algorithms \cite{mcclean16,kandala17,preskill18,bharti21}, which constitute one of the leading candidates for attaining quantum speedups in near-term quantum computers. These algorithms require the estimation of expectation values of a quantum many-body Hamiltonian (encoding, for example, a molecular system relevant for quantum chemistry) on states occurring in the course of a classical-quantum optimisation loop. The Hamiltonian of interest, which in general cannot easily be measured in its eigenbasis, is described by a linear combination of Pauli operators (tensor products of single qubit Pauli matrices). Estimating the expectation values of all relevant Pauli operators, and subsequently the Hamiltonian, involves measuring large collections of incompatible observables. 

To overcome this computational burden, many strategies have been introduced \cite{hamamura20,cade20,torlai20,cai20,huggins21,bonet20,khoda21,huggins21a}, with a typical method involving some grouping of the observables into compatible sets, e.g.,  \cite{jena19,gokhale19,yen20a,verteletskyi20,crawford21,izmaylov19a,zhao20,wu21}. Another approach is a technique called \emph{classical shadows} \cite{huang20}, based on a practical implementation of ideas from shadow tomography \cite{aaronson18}. This framework involves a randomised measurement strategy, implemented with random unitary circuits, that builds a classical approximation of the unknown state to efficiently estimate linear and non-linear functions of the state \cite{huang20}. The protocol readily applies to the problem of estimating multiple expectation values of incompatible observables (as well as many other applications \cite{yu21,elben20,kliesch21,rath21,zhao21,huang21}) and leads to rigorous bounds on the sample complexity (i.e., the number of state preparations) to achieve accurate estimations. 

In this work we present a new approach, conceptually distinct from previous methods for estimating multiple expectation values, using ideas from the theory of measurement incompatibility \cite{busch,heinosaari16,guhne21,heinosaari08}. While joint measurability and commutativity are equivalent notions for projective (von Neumann) measurements, general measurements, i.e., positive operator-valued measures (POVMs), do not necessarily require commutativity to be measured jointly. Importantly, a set of incompatible observables can be measured simultaneously \emph{if} a sufficient amount of noise is added to the measurements. In our scheme we use the randomisation of local projective measurements to simultaneously perform noisy (unsharp) versions of Pauli measurements. This strategy is easily implementable on a quantum computer and can be extended to a locally biased $n$-qubit joint measurement which, after efficient classical post-processing, reproduces the outcome statistics of any sufficiently noisy Pauli measurement (even though no physical noise resides in the system). The outcomes of the joint measurement are then used to construct unbiased estimators of the expectation values of the original noiseless observables, or some linear combination, such as a Hamiltonian (see Fig.\ref{fig:general_idea}).

\begin{figure*}[!t]
\centering
	\includegraphics[width=17cm]{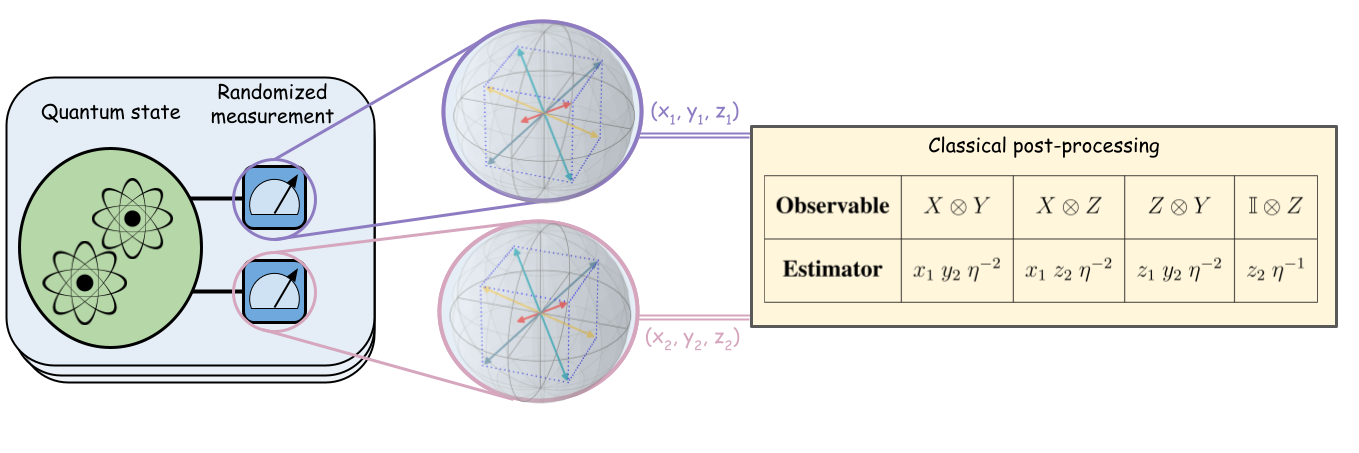}
	\caption{\label{fig:general_idea} Illustration of the main idea of the protocol for simultaneously estimating expectation values of Pauli observables. A joint measurement of noisy Pauli observables (with uniform noise $\eta$) of the form (\ref{eq:qubit_jm_biased}) is implemented on each qubit via randomisation of projective measurements corresponding to pairs of opposite vertices of a cube. Associated with each vertex is a classical outcome $(x_i,y_i,z_i)$ (with $i$ labelling the qubit). The estimators are constructed via classical post-processing by taking a product of outcomes and dividing by the noise. The scheme generalises to a locally biased joint measurement (see main text for details). 
	}
\end{figure*}

Our analysis offers a new and conceptually simpler perspective on understanding randomised protocols for estimating multiple non-commuting observables. We derive bounds on the sample complexity of the protocol, which we show to be identical to those derived for locally biased classical shadows \cite{hadfield20}, a generalised version of the original local classical shadow protocol \cite{huang20}. We then formulate some basic connections between joint measurability and classical shadows, showing that joint measurements can be used to construct classical shadows and, conversely, the shadow protocol defines a joint measurement. Finally, we study the effects of physical noise on our protocol, adapting the joint measurement strategy to minimise the sample complexity when the implementation of physical measurements (used to simulate parent POVMs of unsharp Pauli observables) is affected by readout noise from the quantum computing device \cite{chen19,maciejewski21,bravyi21}. We find that this approach can provide significant performance improvements compared to known generalisations of classical shadows to noisy scenarios.  \cite{koh20,chen21}.

\section{Preliminaries}\label{sec:pre}

Let $\mathcal H$ be a finite dimensional Hilbert space of $\dim \mathcal H=d<\infty$. A measurement is described by a POVM $\M$, with a finite outcome set $\Omega$, and consists of positive semidefinite matrices (effects) $\M(s)\geq 0$, for which $\sum_{s\in \Omega} \M(s)=\id$, where $\id$ is the identity of $\mathcal H$. For a quantum state $\rho$, the outcome probability distribution of $\M$ is given by $p(s|\rho)=\tr{\M(s)\rho}$. A finite collection of measurements $\M_1,\ldots, \M_m$ is \emph{jointly measurable} (compatible) whenever their statistics can be reproduced by classical post-processing of the statistics from a single POVM. In particular, there exists a measurement $\G$ with outcome set $\Omega_{\G}$ such that the effects $\M_j(s_j)$ can be obtained from $\G$ via the classical post-processing (i.e., a stochastic transformation),
\begin{equation}\label{eq:postprocessing}
\M_{j}(s_j)=\sum_{\lambda\in\Omega_{\G}} D(s_j|j,\lambda)\G(\lambda)\,,
\end{equation}
where $0\leq D(s_j|j,\lambda)\leq 1$ and $\sum_{s_j}D(s_j|j,\lambda)=1$, for all $j=1,\ldots,m$. Equivalently, a set of observables is said to be jointly measurable if there exists a single POVM whose marginals yield all effects of the individual observables, otherwise they are said to be incompatible \cite{busch,ali09}. A simple proof of the equivalence is provided in Appendix \ref{A:jm} for completeness.

As an example, consider the qubit Pauli observables $X:=\sigma_x$, $Y:=\sigma_y$ and $Z:=\sigma_z$, and their noisy unsharp versions, $\M^{\eta^x}_X(x)=\frac{1}{2}(\id+x\eta^x X)$, $\M^{\eta^y}_Y(y)=\frac{1}{2}(\id+y\eta^y Y)$ and $\M^{\eta^z}_Z(z)=\frac{1}{2}(\id+z\eta^z Z)$,
with outcomes $x,y,z\in\{\pm 1\}$ and $0\leq\eta^x,\eta^y,\eta^z\leq 1$. It is well known (cf. \cite{heinosaari08,busch86,brougham07}) that the triple is jointly measurable if and only if $({\eta^{x}})^2+({\eta^{y}})^2+({\eta^{z}})^2\leq 1$. The corresponding joint measurement, whose marginals are the unsharp Pauli observables, is 
\begin{equation}\label{eq:qubit_jm_biased}
\G(x,y,z)=\frac{1}{8}\left(\id+x\eta^{x}X+y\eta^{y}Y+z\eta^{z}Z\right)\,.
\end{equation}

Given an arbitrary set of $m$ observables $O_1,\ldots,O_m,$ and an unknown $n$-qubit quantum state $\rho$, our aim will be to provide estimators $\hat{O}_j$ to the expectation values $\tr{O_j\rho}$, up to a certain precision. In particular, we would like to know the number of copies of $\rho$ (i.e., the sample complexity $N$) such that  for all $j=1,\ldots,m$, $|\tr{O_j\rho}-\hat{O}_j|<\epsilon$, with probability at least $1-\delta$.

For simple single shot estimators, such as those for local dichotomic observables, Hoeffding's inequality provides an effective way to bound the sample complexity of the estimation protocol. In other cases, such as a Hamiltonian, the median-of-means approach gives a simple classical post-processing strategy which can reduce the effect of estimation errors \cite{huang20}. This method (explained further in Appendix \ref{A:mom}) depends on the variance of the estimator $\hat{O}_j$ and leads to the following bound on the sample complexity,
\begin{equation}
N=O\left(\frac{\log(m/\delta)}{\epsilon^2}\max_{1\leq j \leq m}\var[\hat{O}_j]\right)\,.
\end{equation}

\section{Estimating expectations via joint measurements}\label{sec:protocol}

The observables we wish to measure simultaneously are the set of $n$-qubit Pauli strings $P=\otimes_{i=1}^n P_{i}$, (with $\mathbb{P}_n$ denoting the set) where $P_{i}\in\{\id, X,Y,Z\}$. The joint measurability strategy to estimate the expectation values $\tr{P\rho}$, for all $P\in\mathbb{P}_n$, can be described succinctly as follows. First we perform a locally biased joint measurement to implement an observable of the form (\ref{eq:qubit_jm_biased}) on each qubit system. To obtain the outcome of the unsharp version of $P$, note that each local measurement provides an outcome tuple $(x_i,y_i,z_i)$, hence we take the product of local outcomes $p_i$ (equal to either $x_i,y_i$ or $z_i$ corresponding to $P_i$). We obtain an unbiased estimator of $\tr{P\rho}$ by dividing $\prod_i p_i$ by the product of local noises (see also Fig. \ref{fig:general_idea}).

Formally, we define the \emph{locally biased joint measurement} $\F({\bf x}_1,\ldots,{\bf x}_n):=\bigotimes_{i=1}^n\G_i({\bf x}_i)$, where each $\G_i$ (with outcomes ${\bf x}_i=(x_i,y_i,z_i)$) is of the form given in Eq. (\ref{eq:qubit_jm_biased}), and the noise parameters $\eta^x_i,\eta^y_i$ and $\eta^z_i$, are biased (and independent) for each qubit. The measurement $\F$ is then a joint measurement for the noisy Pauli measurements
\begin{equation}\label{eq:k_Pauli_noisy}
\M_{P}^{\eta}(s_P)=\frac{1}{2}(\id+s_P\eta_PP)\,,
\end{equation}
with $s_P\in\{\pm 1\}$. The noise coefficient $\eta_P:=\prod_{i\in\text{supp}(P)}\eta^{\nu_i(P)}_i$ is a product of local noises dependent on the individual Pauli operators $P_i$, where $\nu_i(P)=x,y,z$ if $P_i=X,Y,Z$, respectively. For example, if $P=X\otimes \id\otimes Z$, then $\eta_P=\eta^x_1\eta^z_3$. The classical post-processing is defined as
\begin{equation}\label{eq:postprocessing_k_qubit}
D(s_P|P,{\bf x})=
\begin{cases}
1 &\mbox{if} \,\,\, s_P=\mu(P)  \,, \\
0 &\mbox{if} \,\,\, s_P=-\mu(P)\,,
 \end{cases}
 \end{equation}
where $\mu: P\mapsto\prod_{i\in\text{supp}(P)}\mu_i(P)$ is the product of the relevant local outcomes, with $\mu_i(P)=x_i,y_i,z_i$ if $P_i=X,Y,Z$, respectively.

 Importantly, the joint measurement of $(\ref{eq:qubit_jm_biased})$ can be implemented via classical randomisation of qubit projective measurements and is therefore projective simulable \cite{oszmaniec17} (see Appendix \ref{A:sim}). In particular, it can be simulated by a uniform mixture of four qubit projective measurements $\P_j(\pm)=\frac{1}{2}(\id\pm\ve_j\cdot\vsigma)$, where $\ve_1=(\eta^x,\eta^y,\eta^z)$, $\ve_2=(\eta^x,\eta^y,-\eta^z)$, $\ve_3=(\eta^x,-\eta^y,\eta^z)$ and $\ve_4=(\eta^x,-\eta^y,-\eta^z)$, with $({\eta^{x}})^2+({\eta^{y}})^2+({\eta^{z}})^2=1$. Consequently, its implementation does not require additional qubits (as suggested by Naimark's dilation theorem \cite{busch}) and is easily realised on a quantum computer. For the unbiased case $\eta^x=\eta^y=\eta^z=\frac{1}{\sqrt{3}}$, an outcome of $\G$ is obtained by randomly performing one of four projective measurement onto opposite vertices of a cube inscribed in the Bloch sphere.
 
We estimate $\tr{P\rho}$ by sampling from the outputs $s_P$ of the  corresponding unsharp measurement $\M^{\eta}_P$ (from Eq. (\ref{eq:k_Pauli_noisy})). Importantly, for any input state $\rho$, the expectation value of $s_P$ equals $\eta_P\tr{P\rho}$ and therefore it is natural to set an unbiased estimator as $\hat{P}=\frac{1}{\eta_P}s_P$. 

The variance of $\hat{P}$ can be easily upper bounded by $\var[\hat P]=\eta_P^{-2}\left(\mathbb{E}[s_P^2]-\mathbb{E}[s_P]^2\right)\leq \eta^{-2}_P$. If we assume uniform noise, then $\eta_P=\eta^{w(P)}$, with $w(P)=|\text{supp}(P)|$. Applying the optimal single-qubit joint measurement on each qubit, i.e. $\eta=\frac{1}{\sqrt{3}}$, we find $\var[\hat P]\leq 3^{w(P)}$.

\section{Joint measurements of Hamiltonians}\label{sec:hamiltonians}

We now apply the locally biased joint measurement $\F$ to estimate the expectation values $\tr{H\rho}$ of a Hamiltonian $H=\sum_{P\in\mathbb{P}_n}\lambda_PP$, written as a linear combination of Pauli strings, with $\lambda_P\in\mathbb{R}$. From the outcomes of our joint measurement on a $n$-qubit quantum state $\rho$, our single shot estimator of $\tr{H\rho}$ is given by 
\begin{equation}\label{eq:est_ham}
\hat H=\sum_{P\in\mathbb{P}_n}\eta^{-1}_P\lambda_Ps_{P}\,,
\end{equation}
where the outcome $s_P\in\{\pm 1\}$ corresponds to the effect of the POVM defined in (\ref{eq:k_Pauli_noisy}). Clearly, the estimator's expectation satisfies $\mathbb{E}[\hat H]=\sum_P\eta^{-1}_P\lambda_P\mathbb{E}[s_P]=\tr{H\rho}$.

To bound the sample complexity of estimating the expectation value of $H$ via our joint measurement, we obtain the following result.
\begin{proposition}\label{prop:var}
The variance of the estimator $\hat H$ (defined in Eq. (\ref{eq:est_ham})) for the Hamiltonian $H$ is given by
\begin{equation}\label{eq:est_var}
\var[\hat H]=\sum_{P,Q\in\mathbb{P}_n}\frac{\eta_{PQ}f(P,Q)}{\eta_P\eta_Q}\lambda_P\lambda_Q\tr{PQ\rho}-(\tr{H\rho})^2\,,
\end{equation}
where $f(P,Q)=\prod_{i=1}^nf_i(P,Q)$, and
\begin{equation*}
f_i(P,Q)=\begin{cases}
1 &\mbox{if} \,\,\, P_i=\id \,\,\text{or} \,\,Q_i=\id \,\,\text{or}\,\, P_i=Q_i  \,, \\
0 &\mbox{otherwise} \,.
 \end{cases}
\end{equation*}
\end{proposition}
The noise coefficient is defined as $\eta_{PQ}=\prod_{i\in\text{supp}(PQ)}\eta_i^{\nu_i(PQ)}$, where $\nu_i(PQ)$ ignores the phase of $PQ$ and acts as defined earlier. We present the proof in Appendix \ref{A:proof}, which involves calculating $\var[\hat H]=\mathbb{E}[\hat H^2]-(\tr{H\rho})^2$, where $\mathbb{E}[\hat H^2]=\sum_{P,Q\in\mathbb{P}_n}\frac{\lambda_P\lambda_Q}{\eta_P\eta_Q}\mathbb{E}[s_Ps_Q]$. The expectation $\mathbb{E}[s_Ps_Q]$ can be evaluated from the joint measurement $\M^{\eta}_{P,Q}(s_P,s_Q)=\sum_{{\bf x}}D(s_P|P,{\bf x})D(s_Q|Q,{\bf x})\F({\bf x})$. For example, if $P=X\otimes \id$ and $Q=Y\otimes Y$ such that $PQ=XY\otimes Y$, we have, in the uniform noise case, $\M^{\eta}_{P,Q}(s_P,s_Q)=\frac{1}{4}(\id+s_P\eta P+s_Q\eta^2 Q)$, and a simple calculation yields $\mathbb{E}[s_Ps_Q]=0$. On the other hand, if $P=X\otimes \id$ and $Q=X\otimes X$, then $\M^{\eta}_{P,Q}(s_P,s_Q)=\frac{1}{4}(\id+s_P\eta P+s_Q\eta^2 Q+s_Ps_Q\eta PQ)$, and it follows that $\mathbb{E}[s_Ps_Q]=\eta\tr{PQ\rho}$. Thus, the variance of the estimator, and ultimately the performance of the scheme, depends on the two observable marginals of the joint measurement $\F$, for all $P,Q\in\mathbb{P}_n$.

\section{Connections to classical shadows}\label{sec:connections}

A single round of a classical shadow protocol consists of three stages \cite{huang20}. First, a quantum state is transformed via a unitary transformation, $\rho\mapsto U\rho U^{\dagger}$, where $U$ is chosen randomly from an ensemble $\mathcal{U}$. This is followed by a measurement in the computational basis $\{\ket{e}:e\in\{0,1\}^n\}$. Finally, the unitary $U^{\dagger}$ is applied to the post-measurement state $\ket{\hat{e}}$, i.e., $\kb{\hat e}{\hat e}\mapsto U^{\dagger}\kb{\hat e}{\hat e}U$. 

In expectation (over the unitary ensemble and measurement outcomes), this randomised measurement procedure can be described by a quantum channel $\mathcal{M}:\rho\mapsto\mathbb{E}_{U\sim\mathcal{U}}\sum_{e\in\{0,1\}^k}p(e)U^{\dagger}\kb{e}{e}U$, where $p(e)=\bra{e}U\rho U^{\dagger}\ket{e}$. Assuming $\mathcal{M}$ is invertible and that in the $\ell$--th experimental round, $U^{(\ell)}$ is applied and measurement outcome $\ket{\hat e^{(\ell)}}$ is obtained, a \emph{classical snapshot} is defined as $\hat \rho^{(\ell)}=\mathcal{M}^{-1}(U^{\dagger,{(\ell)}}\kb{\hat e^{(\ell)}}{\hat e^{(\ell)}}U^{(\ell)})$. A collection of $N$ such snapshots define the \emph{classical shadow} $S(\rho\,;N)=\{\hat \rho^{(1)},\hat\rho^{(2)},\ldots,\hat\rho^{(N)}\}$. Note that while $\tr{\hat \rho^{(\ell)}}=1$, the snapshot $\hat\rho^{(\ell)}$ is not necessarily positive semidefinite but produces the desired state in expectation, i.e., $\mathbb{E}[\hat \rho]=\mathcal{M}^{-1}(\mathbb{E}[U^{\dagger}\kb{\hat e}{\hat e}U])=\mathcal{M}^{-1}(\mathcal{M}(\rho))=\rho$.

The above formalism can be applied to predict many properties of the quantum state. For example, for any collection of observables $O_1,\ldots,O_m$, the function $\hat O^{sh}_j=\tr{O_j\hat\rho}$ is an unbiased estimator of  $\tr{O_j\rho}$. The sample complexity can be found by bounding the variance of $\hat O^{sh}_j$ with the \emph{shadow norm} \cite{huang20}, i.e., $\var[\hat O^{sh}_j]\leq {\no{O_j}}_{\text{sh}}^2$, as defined in Appendix \ref{A:lbcs}.

Two cases considered in \cite{huang20} construct classical shadows via random local qubit or global $n$-qubit Clifford measurements, and rely on the 3-design property of the Clifford group to compute the shadow norm. In the former case, the measurement procedure is equivalent to performing random Pauli measurements on each qubit, and the classical shadow has the form 
\begin{equation}\label{eq:cs}
\hat \rho=\bigotimes_{i=1}^n\left(3U_i^{\dagger}\kb{\hat e_i}{\hat e_i}U_i-\id\right)\,.
\end{equation}
For an arbitrary Pauli string $P\in\mathbb{P}_n$, the shadow norm is given by $\no{P}_{\text{sh}}^2=3^{w(P)}$.

A modified \emph{locally biased} classical shadow approach for the product Clifford ensemble is described in \cite{hadfield20} and is applied to estimate expectation values of Hamiltonians $H=\sum_{P}\lambda_PP$. Rather than implementing uniformly random Pauli measurements on each qubit, each Pauli $P_i\in\{X,Y,Z\}$ is randomly selected according to a probability distribution $\beta_i(P_i)$. Surprisingly, while the estimator $\hat H^{sh}$ of this protocol differs from $\hat H$ of Eq. (\ref{eq:est_ham}), we observe the following relation between the two.

\begin{observation}\label{obs:lbcs}
The variance of the locally biased classical shadow estimator in \cite{hadfield20} is equivalent to the variance of the joint measurability estimator in Prop. \ref{prop:var}.
\end{observation}

This is shown explicitly in Appendix \ref{A:lbcs}, and requires $(\eta^x_i)^2$, $(\eta^y_i)^2$ and $(\eta^z_i)^2$ to be set as the probabilities $\beta_i(P_i)$ of sampling $X$, $Y$ and $Z$, respectively, for each qubit system. It follows that both approaches have the same sample complexity bounds. For a given Hamiltonian, strategies are provided in \cite{hadfield20,hadfield21} to minimise the variance by optimising over the probabilities $\beta_i(P_i)$. These techniques can also be applied directly to optimise the joint measurement.

We now highlight some further connections when joint measurability is viewed in the framework of classical shadows and vice versa.
\begin{observation}
From the outcomes of the joint measurement $\F$ we can construct a locally biased classical shadow. In the unbiased setting this has a similar form (and the same performance) as the shadow of
Eq. (\ref{eq:cs}).
\end{observation}
For each Pauli string $P\in\mathbb{P}_n$, we construct the product $\mu(P):=\prod_{i\in\text{supp}(P)}\mu_i(P)$ from the outcomes of $\F$. A single shot classical approximation of the quantum state $\rho$ is given by
\begin{equation}
\hat \rho^{\,\text{JM}}=\frac{1}{2^n}\sum_{P\in\mathbb{P}_n}\eta^{-1}_P\mu(P)P=\bigotimes_{i=1}^n\frac{1}{2}(\id+ \ve_i\cdot\vsigma)\,,
\end{equation}
where $\ve_i=(\frac{x_i}{\eta^{x}_i},\frac{y_i}{\eta^{y}_i},\frac{z_i}{\eta^{z}_i})$ and ${\no { \ve_i}}^2=({\eta^{x}_i})^{-2}+({\eta^{y}_i})^{-2}+({\eta^{z}_i})^{-2}$. It follows (see Appendix \ref{A:cs_jm}) that for the $i$-th qubit, we have $\hat \rho^{\,\text{JM}}_i={\no { \ve_i}}\rho_{\tilde \ve_i}+\frac{\id}{2}(1-{\no { \ve_i}})$, where $\tilde\ve_i=\ve_i/\no{\ve_i}$ and $\rho_{\tilde \ve_i}=\frac{1}{2}(\id+\tilde \ve_i\cdot \vsigma)$. If we consider uniform noise such that ${\no { \ve_i}}=\frac{\sqrt{3}}{\eta}$ and take $\eta=\frac{1}{\sqrt{3}}$, the expression simplifies to $\hat \rho^{\,\text{JM}}_i=3\rho_{\tilde \ve_i}-\id$, which has a similar form to the classical shadow (\ref{eq:cs}) but $\rho_{\tilde \ve_i}$ is no longer a Pauli eigenstate. Note also that both shadows give the same variance for the estimators.

\begin{observation}
Any classical shadow defines a joint measurement and provides a sufficient condition for the compatibility of an arbitrary set of measurements.
\end{observation}
Suppose a classical shadow on $\mathcal{H}$ is constructed from a (global) unitary ensemble $\mathcal{U}$ which constitutes a 2-design. The randomised measurement procedure can be described by a single POVM,
$\G(x,U)=\frac{1}{|\mathcal{U}|}U^{\dagger}\kb{x}{x}U$, where $U\in\mathcal{U}$ and $x\in\{0,1\}^n$. As a consequence of the 2-design property, each classical snapshot is given by $\hat \rho_{x,U}=(d+1)U^{\dagger}\kb{x}{x}U-\id$ \cite{huang20}. While $\hat \rho_{x,U}$ is not necessarily positive semidefinite, it has $\tr{\hat\rho_{x,U}}=1$ and satisfies $\mathbb{E}[\hat\rho_{x,U}]=\rho$. For a POVM $\M_j$ we can compute $q(s|j,x,U)=\tr{\M_j(s)\hat\rho_{x,U}}$,
which, in expectation, yields the outcome statistics of the measurement $\M_j$.

To determine which observables can be measured jointly from classical shadows, we require that $q(s|j,x,U)$ is a classical post-processing function, i.e., $\tr{\M_j^{\eta}(s)\hat \rho_{x,U}}\geq 0$, where $\M_j^{\eta}(s)=\eta\M_j(s)+(1-\eta)\frac{\tr{\M_j(s)}}{d}\id$. For the set of all observables this holds if and only if $\eta\leq\frac{1}{d+1}$ (see Appendix \ref{A:cs_jm}). Thus, we can simulate the measurements $\M^{\eta}_j$ for $\eta\leq\frac{1}{d+1}$, from the classical post-processing $q(s|j,x,U)$. Improved bounds can be achieved if $\tr{\M_j(s)U^{\dagger}\kb{x}{x}U}>0$ for all $j$. While the general condition does not give the exact joint measurability region for an arbitrary collection of observables \cite{heinosaari15b}, the classical shadow $\hat\rho^{\,\text{JM}}_i=3\rho_{\tilde \ve_i}-\id$ constructed from the joint measurement of Pauli observables yields the precise incompatibility robustness threshold. Note that classical shadows can also be constructed from informationally complete measurements \cite{atihi21}, for which the above analysis would also apply.

\section{Noisy joint measurements}\label{sec:noise}

The unbiased classical shadow protocol has been adapted to incorporate noise by applying a fixed quantum channel $\Lambda$ to the state after a unitary transformation $U\in\mathcal{U}$, i.e., $U\rho U^{\dagger}\mapsto \Lambda(U\rho U^{\dagger})$, followed by the usual measurement in the computational basis \cite{koh20,chen21}. 
When the unitary ensemble $\mathcal{U}$ is the product Clifford ensemble, and $\Lambda$ describes uncorrelated product noise, the corresponding shadow norm can be calculated explicitly \cite{koh20}. 

This noise model can be equivalently seen as a noisy measurement (described by the dual (unital) CP map  $\Lambda^{*}$ acting on the standard measurement) performed on a perfectly prepared state. Given that the scheme already assumes implementation of perfect product Clifford gates, it is natural to allow implementation of arbitrary noiseless single-qubit unitaries prior to the noisy measurement. We note that if the noise is uncorrelated, one can design unitaries that (on each qubit) transform the noise channel to stochastic (classical) noise (see Appendix~\ref{A:noise_stochastic}).
The resulting single-qubit measurements have the form $\M(0)=\alpha\kb{0}{0}+(1-\beta)\kb{1}{1}$ and $\M(1)=\id-\M(0)$, with $\alpha,\beta \in \left[1/2,1\right]$, which gives a modified single-qubit shadow norm $\no{P}_{\text{sh},\,\Lambda}=(\sqrt{3}/(\alpha+\beta-1))$. 
It is worth noting that stochastic readout noise is indeed one of the main sources of errors in modern quantum devices (see, e.g., \cite{chen19,maciejewski21,bravyi21})

Crucially, our joint measurability scheme can easily incorporate uncorrelated readout noise. To this end, we  find strategies of implementing qubit POVMs via classical randomisation of noisy qubit projective measurements and post-processing. This class of measurements is then used to construct parent POVMs for unsharp versions of Pauli observables. Specifically, we modify the semidefinite program (SDP) derived in \cite{oszmaniec17} which characterises qubit projective simulability to include readout noise (see Appendix \ref{A:noise}), and then incorporate the constraint into the standard joint measurability SDP \cite{guhne21}.

\begin{table}[t!]
\begin{center}
\begin{tabular}{c|ccccc}
\textbf{Encoding / Molecule} & & $\text{H}_2$ & LiH   & Be$\text{H}_2$ & $\text{H}_2$O
\\ \hline
Jordan-Wigner       &  & 1.00    & 0.06 & 0.04           & 0.1          
\\ \hline
Bravyi-Kitaev      &   & 0.13          & 0.78 & 0.55           & 0.61          
\\ \hline
Parity        &  & 0.38          & 0.02 & 0.009           & 0.02          
\end{tabular}
\end{center}
\caption{\label{table_benchmarks}Relative upper bounds on the variance of the estimators for popular quantum chemistry Hamiltonians (with different encodings) in the presence of readout noise, using an optimised joint measurability strategy (see Appendix~\ref{A:optimization_numerics} for details on the optimisation heuristics). The quantities for each molecule and encoding are normalised by the value of the variance upper bound for the noisy classical shadows strategy of the given molecule and encoding (see Appendix~\ref{A:optimization_numbers} for explicit results).
}
\end{table}

We use the above observations to compare the performance of classical shadows and joint measurability schemes when applied to the energy estimation of popular quantum chemistry Hamiltonians in the presence of readout noise.
We note that our methods allow for easily incorporating noise into \emph{biased} strategies, where $\eta^{x,y,z}_i$ are weighted according to the frequency of certain Pauli operators in the Hamiltonian, whereas only \emph{unbiased} noisy classical shadows have been developed so far \cite{koh20,chen21}.
Thus, we can optimise the measurement strategy for each Hamiltonian (using heuristics developed in Appendix~\ref{A:optimization_numerics}).

The results of benchmarks for molecules and encodings \cite{bravyi17,kandala17,hadfield20} are shown in Table~\ref{table_benchmarks}. For each strategy, we calculate the upper bound on the variance (see Eq.~\eqref{eq:est_var} for the joint measurability strategy or corresponding expressions for noisy classical shadows in \cite{chen21}).
As a noise model, we use the noise data from (the most noisy) subsystems of IBM's Washington 126-qubit quantum device. Interestingly, even in the presence of noise, the unbiased joint measurability strategy gives exactly the same performance as classical shadows. In contrast, optimised strategies allow us to obtain a reduction of the variance upper bound by as much as a factor of $\approx 100$.

\

\section{Concluding remarks}\label{sec:conclusion}

In this work we have applied a simple joint measurability strategy, implemented locally, to simultaneously estimate expectation values of collections of incompatible observables. We apply this technique to estimate energies of quantum Hamiltonians and derive bounds on the sample complexity of the protocol. The application of joint measurability to quantum computing problems, as well as the connections to classical shadows, open new research questions to explore. In particular, are there deeper fundamental connections between shadows and joint measurements? Can we gain further insight into the efficiency of computational tasks from the limits of joint measurability, or, conversely, can we construct optimal joint measurements from the performance limits of classical shadows? Furthermore, is joint measurability a useful strategy in other quantum computing applications? In general, this work motivates further studies of measurement incompatibility, especially the characterisation of optimal joint measurement schemes which are projective (or noisy projective) simulable.

\emph{Note added}.---Upon completion of this work we became aware of independent results by Nguyen \emph{et al.} \cite{nguyen22}, which offer a complementary perspective on connecting measurement theory with classical shadows.

\section*{Acknowledgements} 
We thank Ingo Roth for insightful discussions regarding the nature and mathematics of classical shadows. The authors acknowledge financial support from the TEAM-NET project co-financed by the EU within the Smart Growth Operational Program (contract no. POIR.04.04.00-00-17C1/18-00).

\appendix

\section{Joint measurability}\label{A:jm}
We recall two definitions of joint measurability and provide a simple proof that both are equivalent. We also state the definition of incompatibility robustness which quantifies the incompatibility of a set of observables.
\begin{definition}\label{Adef1}
A finite collection of measurements $\M_1,\ldots, \M_m$ is \emph{jointly measurable} (compatible) if there exists a joint measurement $\G$ defined on the product outcome space $\Omega_{\M_1}\times\ldots\times\Omega_{\M_m}$ such that
\begin{eqnarray}\label{Aeq:marginals}
 \sum_{s_2,\ldots,s_m} \G(s_1,\ldots,s_m)&=&\M_{1}(s_1)\,,\nonumber\\
 &\vdots&\nonumber\\
 \sum_{s_1,\ldots,s_{m-1}} \G(s_1,\ldots,s_m)&=&\M_{m}(s_m)\,.
\end{eqnarray}
\end{definition}
Thus, measurements are compatible if their effects constitute the marginals of a single measurement $\G$, otherwise they are \emph{incompatible}. For measurements with finite outcome sets, there are several equivalent formulations of joint measurability \cite{ali09,busch}, one of which incorporates classical post-processing.
\begin{definition}\label{Adef2}
A finite collection of measurements $\M_1,\ldots, \M_m$ is \emph{jointly measurable} if there exists a measurement $\F$ with outcome set $\Omega_{\F}$ such that the effects $\M_j(s_j)$ can be obtained from $\F$ via the classical post-processing
\begin{equation}\label{Aeq:postprocessing}
\M_{j}(s_j)=\sum_{\lambda\in\Omega_{\F}} D(s_j|j,\lambda)\F(\lambda)\,,
\end{equation}
where $0\leq D(s_j|j,\lambda)\leq 1$ and $\sum_{s_j}D(s_j|j,\lambda)=1$, for all $j=1,\ldots,m$. 
\end{definition}
We now provide a simple proof that both definitions are equivalent in the finite dimensional setting.
\begin{proposition}
For finite outcome measurements, the two definitions of joint measurability are equivalent.
\end{proposition}
\begin{proof} For simplicity we show this for a pair of measurements $\M_1$ and $\M_2$, but the proof generalises straightforwardly to any finite collection of measurements. Clearly, the existence of a joint measurement $\G$ such that $\sum_{s_2}\G(s_1,s_2)=\M_1(s_1)$ and $\sum_{s_1}\G(s_1,s_2)=\M_2(s_2)$, describes a classical post-processing of the form (\ref{Aeq:postprocessing}). Conversely, given a measurement $\F$ such that
\begin{eqnarray}
\M_{1}(s_1)&=&\sum_{\lambda\in\Omega_{\F}} D(s_1|1,\lambda)\F(\lambda)\,,\\
\M_{2}(s_2)&=&\sum_{\lambda\in\Omega_{\F}} D(s_2|2,\lambda)\F(\lambda)\,,
\end{eqnarray}
we can define a measurement
\begin{equation}
\tilde \G(s_1,s_2)=\sum_{\lambda\in\Omega_{\F}}D(s_1|1,\lambda)D(s_2|2,\lambda)\F(\lambda)\,,
\end{equation}
with outcome set $\Omega_{\M_1}\times \Omega_{\M_2}$. Setting $D(s_1,s_2|\lambda):=D(s_1|1,\lambda)D(s_2|2,\lambda)$, we have $\sum_{s_1,s_2}D(s_1,s_2|\lambda)=1$. It follows that $\tilde\G$ can be obtained via classical post-processing of $\F$. Furthermore,  $\sum_{s_2}\tilde\G(s_1,s_2)=\M_1(s_1)$ and $\sum_{s_1}\tilde\G(s_1,s_2)=\M_2(s_2)$, therefore $\tilde\G$ is a joint measurement according to Def. \ref{Adef1}.
\end{proof}

The incompatibility of a set of observables can be quantified via the \emph{incompatibility robustness}, defined as the minimum amount of uniform noise needed to render the set jointly measurable \cite{heinosaari15}. This is given by
\begin{equation}
\eta^*=\sup\{\eta>0\,|\,\M^{\eta}_j \text{ are compatible}\,\,\forall \,\alpha\}\,,
\end{equation}
where we consider noisy versions of the original observables,
\begin{equation}\label{Aeq:noisy_obs}
\M_j^{\eta}=\eta \M_j+(1-\eta)\tr{\M_j} \id/d\,.
\end{equation}

\section{The median-of-means estimator}\label{A:mom}

Suppose we perform an experiment $N$ times, and each round results in an (independent) estimator $\hat{O}_j^{(\ell)}$, for each $j=1,\ldots,m$, where $\ell$ denotes the $\ell$--th round of the experiment. We label the resulting set of estimators $\hat{\mathcal{O}}_j=\{\hat{O}_j^{(1)},\ldots,\hat{O}_j^{(N)}\}$, and group its elements into $R$ equally sized parts (each containing $L=\lfloor N/R\rfloor$ elements) to get $R$ estimators
\begin{equation}
\hat O_j^{(r)}=\frac{1}{L}\sum_{i=(r-1)L+1}^{rL}\hat{O}^{(i)}_j\,,
\end{equation}
with $r=1,\ldots,R$. The resulting median-of-means estimator of $\tr{O_j\rho}$, for each $j=1,\ldots,m$, is given by
\begin{equation}
\hat{O}_j(L,R)=\text{median}(\hat O_j^{(1)},\ldots,\hat O_j^{(R)})\,.
\end{equation}
Taking $R=2\log(2m/\delta)$ independent samples of size $L=34\var[O_j]/\epsilon^2$, guarantees
\begin{equation}
|\tr{O_j\rho}-\hat{O}_j(L,R)|\leq \epsilon\,,
\end{equation}
with probability $1-\delta$ \cite{huang20}. Thus, the sample complexity is given by
\begin{equation}
N=O\left(\frac{\log(m/\delta)}{\epsilon^2}\max_{1\leq j \leq m}\var[\hat{O}_j]\right)\,,
\end{equation}
which is bounded by the variance of the estimators.

\section{Randomised measurement implementation}\label{A:sim}

A standard approach for implementing measurements on multiple-qubit quantum computing devices is by a randomised measurement strategy, in which  random unitaries (often belonging to classically tractable families of circuits such as local unitaries, local Cliffords, or global Clifford transformations) are applied to each qubit followed by a projective measurement in the computational basis \cite{elben22}. More generally, measurements that can be implemented using randomisation of projective measurements in some (not necessarily computational) basis, and possible post-processing, are called projective simulable \cite{oszmaniec17}. Such measurements can be written in the form
\begin{equation}\label{Aeq:simulability}
\M(s)=\sum_j p_j(\M)\sum_{s'} D(s|\M,j,s')\P_j(s')\,,
\end{equation}
where $p_j(\M)$ is the probability of implementing the projective measurement $\P_j$, and $D(s|\M,j,s')$ describes the classical post-processing.

The $n$-qubit joint measurement $\F$ defined in the main text can be simulated locally on each qubit by a uniform mixture of the four projective measurements
\begin{equation}\label{Aeq:4projections}
\P_j(\pm)=\frac{1}{2}(\id\pm\ve_j\cdot\vsigma), 
\end{equation}
where $\ve_1=(\eta^x,\eta^y,\eta^z)$, $\ve_2=(\eta^x,\eta^y,-\eta^z)$, $\ve_3=(\eta^x,-\eta^y,\eta^z)$ and $\ve_4=(\eta^x,-\eta^y,-\eta^z)$, with $\vsigma=(X,Y,Z)$. The condition $({\eta^{x}})^2+({\eta^{y}})^2+({\eta^{z}})^2=1$ ensures the measurements are projective, i.e., $\no{\ve_j}=1$. In particular, the single qubit joint measurement $\G$ defined in Eq. (\ref{eq:qubit_jm_biased}), can be written as $\G(\pm j)=\frac{1}{4}\P_j(\pm)$, where the outcome set has been relabelled $\Omega_{\G}=\{\pm j\,|\,j=1,\ldots,4\}$. Clearly, this has the form given in Eq. (\ref{Aeq:simulability}), with $p_j(\G)=1/4$ and
\begin{equation}
D(s|\G,j,s')=
\begin{cases}
1 &\mbox{if} \,\,\, s=js'  \,, \\
0 &\mbox{otherwise} \,,
 \end{cases}
 \end{equation}
for all $j=1,\ldots,4$, with $s'\in\{\pm 1\}$ and $s\in\Omega_{\G}$.
 
Note that for the uniform noise case $\eta=\frac{1}{\sqrt{3}}$, the vectors $\pm\ve_j$ correspond to the eight vertices of a cube centred at $\ze$, as visualised in Fig. \ref{fig:general_idea}.

\section{Proof of Prop. \ref{prop:var}}\label{A:proof}

To estimate the expectation value $\tr{H\rho}$ of a Hamiltonian
\begin{equation}\label{eq:pauli_hamiltonian}
H=\sum_{P\in\mathbb{P}_n}\lambda_PP\,,
\end{equation}
where $\lambda_P\in\mathbb{R}$, our single shot estimator of $\tr{H\rho}$ is given by 
\begin{equation}\label{Aeq:est_ham}
\hat H=\sum_{P}\eta^{-1}_P\lambda_Ps_{P}\,,
\end{equation}
with $s_P\in\{\pm 1\}$ the outcome of the noisy observable 
\begin{equation}
\M_P^{\eta}(s_P)=\frac{1}{2}(\id+s_P\eta_PP)\,.
\end{equation}
It is easily checked that the estimator is unbiased,
\begin{eqnarray}
\mathbb{E}[\hat H]&=&\sum_P\eta^{-1}_P\lambda_P\mathbb{E}[s_P]\\
&=&\sum_P\eta^{-1}_P\lambda_P\sum_{s_P}\tr{\M^{\eta}_P(s_P)\rho}s_P\\
&=&\sum_P\eta^{-1}_P\lambda_P\tr{\eta_P P\rho}=\tr{H\rho}\,,
\end{eqnarray}
where we have taken the expectation over the measurement outcomes $s_P$, with probability $p(s_P|\rho)=\tr{\M^{\eta}_P(s_P)\rho}$.

The variance of the estimator $\hat H$ is given by,
\begin{eqnarray}
\var[\hat H]&=&\mathbb{E}[\hat H^2]-\mathbb{E}[\hat H]^2\\
&=&\mathbb{E}[\hat H^2]-(\tr{H\rho})^2\,,
\end{eqnarray}
with
\begin{eqnarray*}
\mathbb{E}[\hat H^2]=\sum_{P,Q\in\mathbb{P}_n}\frac{\lambda_P\lambda_Q}{\eta_P\eta_Q}\mathbb{E}[s_Ps_Q]\,,
\end{eqnarray*}
and
\begin{eqnarray}\label{Aeq:prod_exp}
\mathbb{E}[s_Ps_Q]=\sum_{s_P,s_Q}\tr{\M^{\eta}_{P,Q}(s_P,s_Q)\rho}s_Ps_Q\,,
\end{eqnarray}
where the outcome products are sampled over the probability distribution $p(s_P,s_Q|\rho)=\tr{\M^{\eta}_{P,Q}(s_P,s_Q)\rho}$. The two observable marginal $\M^{\eta}_{P,Q}$ is obtained from $\F$ via the classical post-processing defined in Eq. (\ref{eq:postprocessing_k_qubit}), such that
\begin{eqnarray}\label{Aeq:joint_pair}
\M^{\eta}_{P,Q}(s_P,s_Q)&=&\sum_{{\bf x}}D(s_P|P,{\bf x})D(s_Q|Q,{\bf x})\F({\bf x})\nonumber\\
&=&\sum_{\substack{{\bf x}:\\\mu(P)\,=\,s_P,\,\mu(Q)\,=\,s_Q}}\F({\bf x})\,.
\end{eqnarray}
Recall from the main text that $\mu: P\mapsto\prod_{i\in\text{supp}(P)}\mu_i(P)$ is a mapping to the product of local outcomes, with $\mu_i(P)=x_i,y_i,z_i$ if $P_i=X,Y,Z$, respectively. For later use we also define $\mu_i(P)=1$ if $P_i=\id$.

In the following lemma we show that $\mathbb{E}[s_Ps_Q]$ depends on the relation between the Pauli operators $P_i$ and $Q_i$ of each subsystem.

\begin{lemma}
For Pauli operators $P,Q\in\mathbb{P}_n$,
\begin{equation}
\mathbb{E}[s_Ps_Q]=\eta_{PQ}f(P,Q)\tr{PQ\rho}\,,
\end{equation}
where $\eta_{PQ}=\prod_{i\in\emph{supp}(PQ)}\eta_i^{\nu_i(PQ)}$,
\begin{equation}
f(P,Q)=\prod_{i=1}^nf_i(P,Q)\,,
\end{equation}
and
\begin{equation*}
f_i(P,Q)=\begin{cases}
1 &\mbox{if} \,\,\, P_i=\id \,\,\text{or} \,\,Q_i=\id \,\,\text{or}\,\, P_i=Q_i  \,, \\
0 &\mbox{otherwise} \,.
 \end{cases}
\end{equation*}
\end{lemma}

\begin{proof}
We first rewrite Eq. (\ref{Aeq:prod_exp}) as
\begin{equation}
\mathbb{E}[s_Ps_Q]=\tr{T_{P,Q}\rho}\,,
\end{equation}
where
\begin{eqnarray}\label{Aeq:f_op}
T_{P,Q}&:=&\sum_{s_P,s_Q}\M^{\eta}_{P,Q}(s_P,s_Q)s_Ps_Q\\
&=&\sum_{{\bf x}_1,\ldots,{\bf x}_n}\mu(P)\mu(Q)\F({\bf x}_1,\ldots,{\bf x}_n)\,.
\end{eqnarray}
The latter equality follows by absorbing the constraints of Eq. (\ref{Aeq:joint_pair}) into the summation. Since
\begin{equation}\label{Aeq:jm_sum}
\F({\bf x}_1,\ldots,{\bf x}_n)=8^{-n}\sum_{R\in\mathbb{P}_n}\eta_R\mu(R)R\,,
\end{equation}
with $\eta_R=\prod_{i\in\text{supp}(R)}\eta_i^{\nu_i(R)}$, we have
\begin{eqnarray}\label{Aeq:F_ip}
T_{P,Q}=8^{-k}\sum_{{\bf x}_1,\ldots,{\bf x}_n}\mu(P)\mu(Q)\sum_{R\in\mathbb{P}_n}\eta_R\mu(R) R\,.
\end{eqnarray}

The next lemma shows that the operator $T_{P,Q}$ depends on whether $P_i$ and $Q_i$ qubit-wise commute. Consequently, Lemma \ref{prop:var} follows directly from this observation.
\end{proof}
\begin{lemma}\label{Alem:F}
For the operator $T_{P,Q}$ defined in Eq. (\ref{Aeq:f_op}), the following conditions hold:
\begin{itemize}
\item[1.] If $P_i=\id$ or $Q_i=\id$ or $P_i=Q_i$ for all $i\in\{1,2,\ldots,n\}$, then $T_{P,Q}=\eta_{PQ}PQ$.
\item[2.] If $Q_i,P_i\neq\id$ and $Q_i\neq P_i$ for at least one $i\in\{1,2,\ldots,n\}$, then $T_{P,Q}= 0$.
\end{itemize}
\end{lemma}
\begin{proof}
Since the operator $T_{P,Q}$ in Eq. (\ref{Aeq:F_ip}) involves a summation over ${\bf x}_i=(x_i,y_i,z_i)$ for all $i\in\{1,\ldots,n\}$, where $x_i,y_i,z_i\in\{\pm 1\}$, the only $n$-qubit Pauli operator $R$ that contributes to $T_{P,Q}$ must satisfy
\begin{equation}
\mu(R)=\mu(P)\mu(Q)\,.
\end{equation}

In particular, $\mu_i(R)=\mu_i(P)\mu_i(Q)$ for all $i\in\{1,2,\ldots,n\}$. If $\mu_i(R)=1$ (i.e. $R_i=\id$), then $\mu_i(Q)=\mu_i(P)$. If $\mu_i(R)\in\{x_i,y_i,z_i\}$ then either $\mu_i(Q)=\mu_i(R)$ and $\mu_i(P)=1$, or $\mu_i(Q)=\mu_i(R)$ and $\mu_i(Q)=1$. It follows that $R=PQ$ and $\eta_R=\eta_{PQ}$. This is equivalent to part 1 of Lemma \ref{Alem:F}.

For the second part, the product $\mu(P)\mu(Q)$ contains at least one pair $\mu_i(P)\mu_i(Q)$ such that $\mu_i(P)\neq \mu_i(Q)$ and $\mu_i(P)\mu_i(Q)\neq1$, e.g., $\mu_i(P)\mu_i(Q)=x_iy_i$. Hence, $\mu_i(R)\neq \mu_i(P)\mu_i(Q)$, and the corresponding $R$, when summed over all ${\bf x}_i$, does not appear in $T_{P,Q}$. It follows that $T_{P,Q}=0$.

\end{proof}

To highlight the distinct cases we consider two examples; the first when $P$ and $Q$ are chosen such that $\mathbb{E}[s_Ps_Q]=0$, and the second when $\mathbb{E}[s_Ps_Q]\neq0$. In both cases we consider two-qubit strings $P=P_1\otimes P_2$ and $Q=Q_1\otimes Q_2$, and a joint measurement with uniform noise. We will also assume $P=X\otimes \id$.

\

{\bf Case 1:} Let $Q_1=Q_2=Y$, such that $PQ=XY\otimes Y$. Importantly, $Q_1\notin\{\id,X\}$. Evaluating Eq. (\ref{Aeq:joint_pair}), with $\F$ written in the form of Eq. (\ref{Aeq:jm_sum}), and for outcomes $s_P=s_Q=1$, we find
\begin{eqnarray*}
\M^{\eta}_{P,Q}(1,1)&=&\frac{1}{8^2}\sum_{\substack{{\bf x}_1,{\bf x}_2:\\\,x_1=\,1,\,y_1y_2\,=\,1}}\sum_{R\in\mathbb{P}_n}\mu(R)\eta^{w(R)}R\\
&=&\frac{1}{8^2}\sum_{\substack{{\bf x}_1,{\bf x}_2:\\\,x_1=\,1\\y_1y_2\,=\,1}}(\id+\eta x_1P+\eta^2y_2(y_1Q- iz_1PQ))\\
&=&\frac{1}{4}(\id+\eta P+\eta^2 Q)\,.
\end{eqnarray*}
Note that only four operators appear in the second equality since all other terms sum to zero. The terms involving $PQ$ also sum to zero since the coefficients $z_1y_2\in\{\pm 1\}$ are not restricted by the constraints. The joint measurement is given by
\begin{equation}
\M^{\eta}_{P,Q}(s_P,s_Q)=\frac{1}{4}(\id+s_P\eta P+s_Q\eta^2 Q)\,.
\end{equation}
A simple calculation yields $\mathbb{E}[s_Ps_Q]=0$. 

\

{\bf Case 2:} Let $Q_1=Q_2=X$, such that $PQ=\id\otimes X$. It follows that
\begin{eqnarray*}
\M^{\eta}_{P,Q}(1,1)&=&\sum_{\substack{{\bf x}_1,{\bf x}_2:\\\,x_1=\,1,\,x_1x_2\,=\,1}}\sum_{R\in\mathbb{P}_n}\mu(R)\eta^{w(R)}R\\
&=&\frac{1}{8^2}\sum_{\substack{{\bf x}_1,{\bf x}_2:\\\,x_1=\,1\\x_1x_2\,=\,1}}(\id+\eta x_1P+\eta^2x_1x_2Q+\eta x_2 PQ)\\
&=&\frac{1}{4}(\id+\eta P+\eta^2 Q+\eta PQ)\,.
\end{eqnarray*}
In contrast to case 1, the terms involving the operator $PQ$ no longer sum to zero since the coefficient $x_2$ is restricted by the constraints $x_1=1$ and $x_1x_2=1$. Note that, in general, the constraints on $\mu(P)$ and $\mu(Q)$ restrict the coefficients of $PQ$ when $Q_i=P_i$ or $P_i=\id$ or $Q_i=\id$ for all $i$. For our specific case, the measurement is given by
\begin{equation}
\M^{\eta}_{P,Q}(s_P,s_Q)=\frac{1}{4}(\id+s_P\eta P+s_Q\eta^2 Q+s_Ps_Q\eta PQ)\,,
\end{equation}
and it follows that $\mathbb{E}[s_Ps_Q]=\eta\tr{PQ\rho}$.

\section{Comparison of locally biased classical shadow variance}\label{A:lbcs}

The sample complexity of the classical shadow estimation procedure described in the main text can be found by bounding the variance of $\hat O^{sh}=\tr{O\hat\rho}$ with the \emph{shadow norm} \cite{huang20},
\begin{equation}\label{Aeq:shadow_bound}
\var[\hat O^{sh}]\leq \no{O}_{\text{sh}}^2\,,
\end{equation}
where
\begin{equation}
\no{O}_{\text{sh}}=\max_{\sigma\in\sh}\sqrt{\underset{U\sim\mathcal{U}}{\mathbb{E}}\sum_{e\in\{0,1\}^n}p(e)\bra{e}U\mathcal{M}^{-1,\dagger}(O)U^{\dagger}\ket{e}^2},
\end{equation}
with $p(e)=\bra{e}U\sigma U^{\dagger}\ket{e}$, and $\sh$ the set of states on $\mathcal{H}$.

The classical shadow protocol provides a method for estimating expectation values of Hamiltonians of the form given in Eq.  (\ref{eq:pauli_hamiltonian}). As described in \cite{huang20}, a Pauli measurement $P_i\in\{X,Y,Z\}$ is uniformly at random performed on the $i$--th qubit of the quantum state, for each $i=1,\ldots,n$. A non-zero contribution in the estimator of $\tr{H\rho}$ appears for each Pauli string $Q$ (of the Hamiltonian) which qubit-wise commutes with the measured $P$. The estimator ignores Pauli strings in the Hamiltonian that do not qubit-wise commute with $P$.

A generalised protocol, introduced in \cite{hadfield20}, allows for locally biased measurements of Pauli observables. On the $i$--th qubit, a Pauli measurement $P_i\in\{X,Y,Z\}$ is randomly selected according to the probability distribution $\beta_i(P_i)$, such that $P\in\mathbb{P}_n$ is measured with probability $\beta(P)=\Pi_i\beta_i(P_i)$. The uniform case (cf. \cite{huang20}) is recovered with $\beta_i(P_i)=\frac{1}{3}$ for each $P_i\in\{X,Y,Z\}$.

The estimator of $\tr{H\rho}$ for the Hamiltonian defined in Eq. (\ref{eq:pauli_hamiltonian}) is
\begin{equation}
\hat H^{sh}=\sum_{Q\in\mathbb{P}_n}\lambda_Qg(P,Q)\zeta(P,\text{supp}(Q))\,,
\end{equation}
where $\zeta(P,\text{supp}(Q))=\Pi_{i\in \text{supp}(Q)}\zeta(P,i)$ and $\zeta(P,i)\in\{\pm 1\}$ is the eigenvalue associated with measurement of $P_i$. Furthermore,
\begin{equation}
g(P,Q)=\prod_{i=1}^ng_i(P,Q)\,,
\end{equation}
and
\begin{equation*}
g_i(P,Q)=\begin{cases}
1 &\mbox{if} \,\,\, P_i=\id \,\,\text{or} \,\,Q_i=\id \,, \\
(\beta_i(P_i))^{-1} & \mbox{if} \,\,\, P_i=Q_i\neq \id \,, \\
0 &\mbox{otherwise} \,.
 \end{cases}
\end{equation*}

As shown in \cite{hadfield20}, the variance of the estimator is given by
\begin{equation}\label{Aeq:sh_var}
\var[\hat H^{sh}]=\sum_{P,Q}g(P,Q)\lambda_P\lambda_Q\tr{PQ\rho}-(\tr{H\rho})^2\,.
\end{equation}
Surprisingly, the variance of the joint measurement estimator (Prop. \ref{prop:var}) and classical shadow estimator (Eq. (\ref{Aeq:sh_var}))are identical if the noise parameter $\eta_P$ and probability distribution $\beta(P)$ satisfy
\begin{equation}
\beta_i(P_i)=(\eta_i^{\nu_i(P)})^2\,,
\end{equation}
with $\nu_i(P)$ defined in the main text and $P_i\in\{X,Y,Z\}$.

One easily checks that
\begin{equation}
g(P,Q)=\frac{\eta_{PQ}f(P,Q)}{\eta_P\eta_Q}\,,
\end{equation}
with $f(P,Q)$ defined in Prop. \ref{prop:var}.
Suppose $P_i=\id$ or $Q_i=\id$, such that $g(P_i,Q_i)=1$. Then, $\eta_{P_iQ_i}/(\eta_{P_i}\eta_{Q_i})=1$. On the other hand, if $P_i=Q_i\neq\id$, we have $g(P_i,Q_i)=(\beta_i(P_i))^{-1}$ and  $\eta_{P_iQ_i}/(\eta_{P_i}\eta_{Q_i})=(\eta_i^{\nu_i(P)})^{-2}$. Note that for the unbiased case (i.e., with uniform noise) $\eta=\frac{1}{\sqrt{3}}$, we recover the variance of the standard classical shadow protocol, since $(\beta_i(P_i))=\eta^2=\frac{1}{3}$.

\section{Classical shadows from joint measurements and vice versa}\label{A:cs_jm}

We now show that joint measurements can be used to construct classical shadows, and, conversely, the classical shadow protocol defines a joint measurement and provides a sufficient condition for the compatibility of an arbitrary set of measurements.

\begin{observation*}
From the outcomes of the joint measurement $\F$ we can construct a classical shadow which has a similar form to the shadow given in Eq. (\ref{eq:cs}).
\end{observation*}

Suppose ${\bf x}=({\bf x}_1,\ldots,{\bf x}_n)$ is the outcome of a single joint measurement $\F$ on a state $\rho$, where ${\bf x}_i=(x_i,y_i,z_i)$ labels the outcome associated with the $i$--th qubit system. For each Pauli string $P\in\mathbb{P}_n$, we construct the product $\mu(P):=\prod_{i\in\text{supp(P)}}\mu_i(P)$, which provides a single shot estimator of the expectation value of $P$, since
\begin{equation}
\underset{p({\bf x}|\rho)}{\mathbb{E}}\mu(P)=\eta_P\tr{P\rho}\,,
\end{equation}
where $p({\bf x}|\rho)=\tr{\F({\bf x})\rho}$. It follows that we can construct a single shot classical approximation of the quantum state $\rho$ via
\begin{eqnarray}
\hat \rho^{\,\text{JM}}&=&\frac{1}{2^n}\sum_{P\in\mathbb{P}_n}\frac{1}{\eta_P}\mu(P)P\\
&=&\bigotimes_{i=1}^n\frac{1}{2}(\id+ \ve_i\cdot\vsigma)\,,
\end{eqnarray}
where $\ve_i=(\frac{x_i}{\eta^{x}_i},\frac{y_i}{\eta^{y}_i},\frac{z_i}{\eta^{z}_i})$ and ${\no { \ve_i}}^2=({\eta^{x}_i})^{-2}+({\eta^{y}_i})^{-2}+({\eta^{z}_i})^{-2}$. On the $i$--th qubit,
\begin{eqnarray}
\hat \rho^{\,\text{JM}}_i&=&\frac{1}{2}(\id+ \ve_i\cdot\vsigma)\\
&=&\frac{1}{2}(\id+{\no {\ve_i}}\tilde \ve_i\cdot \vsigma)\\
&=&\no {\ve_i}\rho_{\tilde\ve_i}+\frac{\id}{2}(1-\no{\ve_i})\,,
\end{eqnarray}
with $\tilde\ve_i=\ve_i/\no{\ve_i}$ and $\rho_{\tilde \ve_i}=\frac{1}{2}(\id+\tilde \ve_i\cdot \vsigma)$. 

This can be viewed as a locally biased version of a classical shadow, which we obtain from the joint measurement $\F$. In fact, if we restrict to the unbiased form of $\F$, with uniform noise $\eta=\frac{1}{\sqrt{3}}$, then  ${\no { \ve_i}}=\frac{\sqrt{3}}{\eta}=3$, and we have
\begin{equation}\label{Aeq:qubit_cs}
\hat \rho^{\,\text{JM}}_i=3\rho_{\tilde \ve_i}-\id\,.
\end{equation}
This has a similar form to the classical shadow defined in Eq. (\ref{eq:cs}) but $\rho_{\tilde \ve_i}$ is no longer a Pauli eigenstate. As the variance of the estimators for both approaches are identical (see Appendix \ref{A:lbcs}), the two shadows provide the same performance.

\begin{observation*}
A classical shadow defines a joint measurement and provides a sufficient condition for the compatibility of an arbitrary set of measurements.
\end{observation*}

Suppose a classical shadow on $\mathcal{H}$ is constructed from a global unitary ensemble $\mathcal{U}$ which constitutes a 2-design. The randomised measurement procedure can be described by a single POVM,
\begin{equation}
\G(x,U)=\frac{1}{|\mathcal{U}|}U^{\dagger}\kb{x}{x}U,
\end{equation}
where $U\in\mathcal{U}$ and $x\in\{0,1\}^n$. It follows from the 2-design property of $\mathcal{U}$ that each classical snapshot has the form \cite{huang20}:
\begin{equation}
\hat \rho_{x,U}=(d+1)U^{\dagger}\kb{x}{x}U-\id\,.
\end{equation}
While $\hat \rho_{x,U}$ is not necessarily positive semidefinite, it has $\tr{\rho_{x,U}}=1$ and satisfies $\mathbb{E}[\hat\rho_{x,U}]=\rho$. For a POVM $\M_j$ we can compute
\begin{equation}
q(s|j,x,U)=\tr{\M_j(s)\hat\rho_{x,U}}\,,
\end{equation}
which, in expectation, yields the outcome statistics of the measurement $\M_j$. To determine which observables can be measured jointly from classical shadows, we require that $q(s|j,x,U)$ is a classical post-processing function, i.e.,
\begin{equation}\label{eq:cpp_condition}
\tr{\M_j^{\eta}(s)\hat \rho_{x,U}}\geq 0\,,
\end{equation}
where
\begin{equation}
\M_j^{\eta}(s)=\eta\M_j(s)+(1-\eta)\frac{\tr{\M_j(s)}}{d}\id\,.
\end{equation}
This holds if, for all $x,U$ and $s$,
\begin{equation*}
\left(d^{-1}(1-\eta)-\eta\right)\tr{\M_j(s)}+\eta(d+1)\tr{\M_j(s)\varphi_x^U}\geq0\,,
\end{equation*}
where $\varphi_x^U=U^{\dagger}\kb{x}{x}U$. Since $\tr{\M_j(s)\varphi_x^U}\geq0$, it follows that
\begin{equation}
\eta\leq\frac{1}{d+1}\,.
\end{equation}
Thus, we can simulate the measurements $\M^{\eta}_j$ for $\eta\leq\frac{1}{d+1}$, from the classical post-processing $q(s|j,x,U)$. Improved bounds can be achieved if $\tr{\M_j(s)\varphi_{x,U}}>0$. 

For the case of the qubit classical shadow of Eq. (\ref{Aeq:qubit_cs}), given by $\hat\rho^{\,\text{JM}}_{i}=3\rho_{\tilde \ve_i}-\id$, we can apply the condition of Eq. (\ref{eq:cpp_condition}),
\begin{equation}
\tr{\M^{\eta}_P(s)\hat\rho^{\,\text{JM}}_{i}}=\tr{\frac{1}{2}(\id+s\eta P)\hat\rho^{\,\text{JM}}_{i}}\geq 0\,,
\end{equation}
where $P\in\{X,Y,Z\}$. After some simple algebra we get the precise joint measurability region $\eta\leq\frac{1}{\sqrt{3}}$.

\section{Transforming single-qubit measurement noise to stochastic noise}\label{A:noise_stochastic}

In the main text we stated that access to noiseless single-qubit unitaries prior to noisy measurements allows for the transformation of measurement (readout) noise to stochastic (classical) noise.
This observation was made in ``further research'' discussions in \cite{maciejewski20} in the context of error-mitigation.
To describe how this works, first note that an arbitrary single-qubit measurement $\M$ is specified by a single measurement effect $\M(0) = \alpha \psi + (1-\beta) \psi^{\perp}$, where $\psi,\ \psi^{\perp}$ are pure quantum states, and $\braket{\psi}{\psi^{\perp}}=0$. 
From the normalisation, it follows that $\M(1) = \id-\M(0) = (1-\alpha) \psi+\beta \psi^{\perp}$, hence both effects commute.
This ensures the existence of a unitary $U$ that diagonalises both effects. Consequently, we can obtain a transformed POVM $\tilde{\M}$ with effects $\tilde{\M}(0) = U\M(0)U^{\dagger} = \alpha \ketbra{0}{0}+(1-\beta)\ketbra{1}{1}$ and $\tilde{\M}(1) = (1-\alpha) \ketbra{0}{0} + \beta \ketbra{1}{1}$. 

It is straightforward to verify that such a transformed measurement is related to the computational basis measurement via a stochastic transformation $T$, i.e., we have $\tilde{\M}(i) = \sum_{j=0,1} T_{ij}\ketbra{j}{j}$, or in matrix form
\begin{align}\label{Aeq:sm}
    \begin{pmatrix}
\tilde{\M}(0)\\
\tilde{\M}(1)
\end{pmatrix}
 =
\underbrace{\begin{pmatrix} 
\alpha & 1-\beta \\
1-\alpha & \beta 
\end{pmatrix}}_{T}
\begin{pmatrix}
\ketbra{0}{0}\\
\ketbra{1}{1}
\end{pmatrix}\,.
\end{align}
This can be viewed as noise that applies classical post-processing to the ideal measurement outcomes \cite{maciejewski20}. The parameters $\alpha$ and $\beta$ can be interpreted as the probability of measuring outcomes $0$ and $1$ if the input is $\ketbra{0}{0}$ and $\ketbra{1}{1}$, respectively.
Note that, for example, the depolarising, amplitude damping, and Pauli noise channels all act as stochastic maps on the computational basis measurement without requiring an additional unitary transformation.

\section{Implementation of a joint measurement via noisy projective measurements}\label{A:noise}

We start by providing a semidefinite program (SDP) for characterising qubit measurements which can be simulated by noisy projective measurements. By noisy projective measurements we mean projective measurements affected by classical readout noise (cf. the preceeding section). In what follows we shall refer to measurements simulated by noisy projective measurements as noisy projective simulable. Let $\mathcal{M}_k$ denote the set of  POVMs on $\mathbb{C}^d$ with $k\geq d$ outcomes, and let $\mathcal{P}$ denote the set of qudit projective measurements. A measurement $\M\in\mathcal{M}_k$ is projective simulable (which we denote by the set $\mathbb{S}_k$) if it can be realised by classical randomisation followed by classical post-processing of measurements from $\mathcal{P}\subset \mathcal{M}_k$ \cite{oszmaniec17} (see also Appendix \ref{A:sim}).

For qubit systems, a measurement $\M\in\mathcal{M}_k$ is projective simulable if and only if it can be expressed as a convex combination of two-outcome POVMs. This leads to the formulation of a SDP where $\M$ can be written as a sum of $k\choose 2$ subnormalised POVMs $\N_{ij}$, with $i,j\in\{1,\ldots, k\}$ and $i<j$, such that the effects $\N_{ij}(s)$ are non-zero only when $s\in\{i,j\}$. Formally, the SDP reads:
\begin{eqnarray*}
&&\text{given}\quad \M(s)\,\,\, \text{for all} \,\,\,s=1,\ldots,k\\
&&\text{check whether}\\
&&\M(s)=\sum_{\substack{i,j ,\ i<j}}\ \N_{ij}(s)\\
&&\text{subject to}\\
&&\N_{ij}(s)\geq 0,\,\,s=1,\ldots,k,\\
&&\N_{ij}(s)=0 \,\,\,\text{for}\,\,\, s\notin \{i,j\},\\
&&\sum_{s}\N_{ij}(s)=p_{ij}\id\,, \,\,\,p_{ij}\geq 0,\\ 
&&\sum_{i\neq j}p_{ij}=1\, .
\end{eqnarray*}

In order to incorporate noisy qubit projective simulability, the optimisation variables (POVMs) $\N_{ij}$ are transformed under the action of stochastic noise.
Specifically, the noise for each POVM is specified by a stochastic map $\Lambda_{\text{noise}}^{ij}$ that transforms the two-outcome measurements as  $\Lambda^{ij}_{\text{noise}}(\N_{ij})(k) = \sum_{l\in\left\{i,j\right\}}(\Lambda^{ij}_{\text{noise}})_{kl}\N(l)$. Furthermore, the presence of  classical noise breaks the symmetry between the outcomes of binary measurements which makes it necessary to incorporate strategies that exchange the role of the two outputs: $\text{T}^{ij}_{\text{bitflip}}(\N_{ij})(s)$ swaps outcomes $i$ and $j$, i.e., we have $\text{T}^{ij}_{\text{bitflip}}(\N_{ij})(i) = \N_{ij}(j)$ and $\text{T}^{ij}_{\text{bitflip}}(\N_{ij})(j) = \N_{ij}(i)$. Finally, the presence of classical noise in general makes it impossible to implement deterministic dichotomic measurements (which in the noiseless scenarios are realised by projective measurements of the form $\N_{ij}(i)=0,\N_{ij}(j)=\id$). Of course, such measurements can be realised in the noisy scenario by deterministic post-processing. Incorporating these remarks yields the following SDP for characterising $\mathbb{S}^{\Lambda}_k$, i.e., the set of  noisy qubit  projective simulable measurements with $k$ outcomes:

\begin{eqnarray*}
&&\text{given}\quad \M(s)\,\,\, \text{for all} \,\,\,s=1,\ldots,k\\
&&\text{check whether}\\
&&\M(s)=\sum_{\substack{i,j ,\ i<j}}\ \tilde{\N}_{ij}(s)+\sum_{\substack{i,j ,\ i>j}}\text{T}^{ij}_{\text{bitflip}}(\tilde{\N}_{ij})(s) +q_s\id\\
&&\text{subject to}\\
&& \tilde{\N}_{ij} = \Lambda^{ij}_{\text{noise}}(\N_{ij})\,, \\
&&\N_{ij}(s)\geq 0,\,\,s=1,\ldots,k,\\
&&\N_{ij}(s)=0 \,\,\,\text{for}\,\,\, s\notin \{i,j\}\,,\\
&&\sum_{s}\N_{ij}(s)=p_{ij}\id\,, \,\,\,p_{ij}\geq 0,\ q_{s}\geq 0\,,\,\,\, \\
&&\sum_{i\neq j}p_{ij}+\sum_{s}q_s=1\,.
\end{eqnarray*}

We now want to find the optimal joint measurement $\E\in\mathbb{S}^{\Lambda}_8$ with marginals that yield the set of unsharp Pauli measurements  $\M^{\eta}_X(x)=\frac{1}{2}(\id+  x\eta X)$, $\M^{\eta}_Y(y)=\frac{1}{2}(\id+y\eta Y)$ and $\M^{\eta}_Z(z)=\frac{1}{2}(\id+z\eta Z)$, for a particular noise matrix $\Lambda$. Incorporating this constraint into the standard joint measurability SDP we get:
\begin{eqnarray*}
&&\text{maximise}\quad \eta\\
&&\text{over}\quad \eta\,,\,\,\E\\
&&\text{subject to}\\
&&\E\in\mathbb{S}^{\Lambda}_8\,,\\
&&\E({\bf x})\geq0\,\,\, \text{for all}\,\, {\bf x}\in\Omega_{\E} \,,\,\,\,\eta\leq 1\,,\\
&&\M^{\eta}_X(x)=\sum_{y,z}\E({\bf x})\,,\ldots,\M^{\eta}_Z(z)=\sum_{x,y}\E({\bf x})\,.\\
\end{eqnarray*}

\section{Heuristics for optimising joint measurability in the presence of noise}\label{A:optimization_numerics}

In Table~\ref{table_benchmarks} we presented results from the calculations of the upper bound on the variance, obtained for different estimators of energies of molecular Hamiltonians in the presence of readout noise.
We now discuss the heuristic methods we used to optimise the joint measurability (JM) strategy for each Hamiltonian.

We start by recalling the expression for the variance of the JM estimator (Eq.~\eqref{eq:est_var} in the main text) and derive explicit upper bounds,
\begin{align*}\label{appendix_eq:variance_jm}
\var[\hat H]=&\sum_{P,Q\in\mathbb{P}_n}\frac{\eta_{PQ}f(P,Q)}{\eta_P\eta_Q}\lambda_P\lambda_Q\tr{PQ\rho}-(\tr{H\rho})^2 \\ 
&\leq \sum_{P,Q\in\mathbb{P}_n}\frac{\eta_{PQ}f(P,Q)}{\eta_P\eta_Q}\lambda_P\lambda_Q\tr{PQ\rho}\ . \numberthis
\end{align*}
Note that the above bound simply disregards the non-positive term that is not dependent on the measurement strategy.
For ease of notation, let us define
\begin{equation}
    \kappa_{P,Q} \coloneqq \frac{\eta_{PQ}f(P,Q)}{\eta_P\eta_Q}\ \lambda_{P}\lambda_{Q}\ .
\end{equation}
Now we write 
\begin{align*}\label{appendix_eq:jm_norm}
\sum_{P,Q\in\mathbb{P}_n}\ \kappa_{P,Q} \tr{PQ\rho} &=\tr{\left(\sum_{P,Q\in\mathbb{P}_n}\kappa_{P,Q}PQ\right)\ \rho}  \\
&\leq \max_{\rho}\tr{\left(\sum_{P,Q\in\mathbb{P}_n}\kappa_{P,Q}PQ\right)\ \rho} \\
&=\no{\sum_{P,Q\in\mathbb{P}_n}\kappa_{P,Q} PQ }_{\infty} \eqqcolon \no{H}^2_{\text{JM}} , \numberthis
\end{align*}
where $\no{.}_{\infty}$ denotes the operator norm.
Hence, the above bound corresponds to the worst-case among all states from the term in the variance that is dependent on the measurement strategy.
We note that, perhaps unsurprisingly, expressions with exactly the same structure appear in the classical shadows framework and correspond to the shadow norm \cite{huang20,koh20}.
We therefore use the notation $\no{H}^{2}_{\text{JM}}$ to denote analogously the ``joint measurability'' (or JM-)norm of a Hamiltonian corresponding to the specific joint measurability strategy.

We are interested in finding a joint measurement strategy that minimises the above analogue of the shadow norm. We minimise the JM-norm via standard black-box optimisation methods, where the visibilities (i.e., the sharpness) of each qubit are the optimisation variables.
To perform this optimisation one needs to:
\begin{enumerate}
    \item[(a)] specify a cost function; 
    \item[(b)] ensure that the optimal visibilities guarantee joint measurability in the presence of noise;
    \item[(c)] ensure that the parent POVM has the structure of Eq.~\eqref{eq:qubit_jm_biased}.
\end{enumerate}

We now discuss all points in detail. We start by performing an optimisation that does not necessarily guarantee the desired structure of the parent POVM (point (c)), and leave this issue for later. The straightforward choice for (a) is to treat the JM-norm itself as a cost function. Unfortunately, this requires calculating operator norms of exponentially large (and generally not sparse) operators, and is thus not practical. 
Another possibility would be to choose a cost function as $\sum_{P,Q}|\kappa_{P,Q}|$ (note that this corresponds to bounding the norm of the sum via the sum of the norms).
While this appears to be a reasonable choice, each function evaluation for this optimisation would require performing a number of elementary operations quadratic in the number of Hamiltonian terms.
Since the Hamiltonians considered usually consist of thousands of terms, this is relatively costly to optimise in practice (it is, however, not prohibitive). We therefore choose a simplified cost function given by
\begin{align}
    \text{cost}_{\text{diag}}(H) = \sum_{P}\kappa_{P,P} \ ,
\end{align}
which we call the ``diagonal'' cost function in analogy to a similar strategy considered in the (noiseless) locally biased classical shadows (LBCS) settings \cite{hadfield20}.
Indeed, we note that the above optimisation is analogous to the problem considered in the LBCS paper where the authors aim to optimise the classical shadow strategy by introducing biases that minimise some cost function related to the Hamiltonian in question.

Now we discuss point (b).
First, the usual joint measurability constraints $\eta_{i}^{x,y,z} \in \left[0,1\right]$ and  $(\eta_{i}^{x})^2+(\eta_{i}^{y})^2+(\eta_{i}^{z})^2\leq 1$ can be directly incorporated into the optimisation procedure provided that the used optimisers allow for non-linear constraints.
However, one needs to ensure that the chosen visibilities admit a joint measurability strategy via classical randomisation (and post-processing) of \emph{noisy} two-outcome qubit measurements.
One way to approach this is to introduce an additional penalty term to the cost function that penalizes the visiblities for which such a strategy does not exist.
However, this requires solving an SDP for each qubit at each function evaluation. While the SDPs are generally solved very quickly, in the course of the optimisation the cost function is to be evaluated tens of thousands of times, which can introduce a significant practical overhead.
We therefore implement a different strategy that does not require solving SDPs at each function evaluation.
Namely, we first perform a black-box optimisation of the visibilities without checking whether the JM-strategy exists (or, equivalently, whether the relevant SDP is feasible). 
We then take the optimal visibilities $\tilde{\eta}_{i}^{x}, \tilde{\eta}_{i}^{y}, \tilde{\eta}_{i}^{z}$ for each qubit $i$, and check whether the SDP is feasible. In cases when the answer is negative, we solve an additional SDP that finds the highest multiplicative factor $r^{*}$, such that for $r^{*}\tilde{\eta}_{i}^{x}, r^{*}\tilde{\eta}_{i}^{y}, r^{*}\tilde{\eta}_{i}^{z}$, the joint measurability strategy with noisy measurements \emph{does} exist.

To deal with point (c), we simply check if the parent POVM in the solution from (b) has the desired structure; if not, we allow for another multiplicative factor to increase the noise and enforce it.
Notably, we find that POVMs obtained in (b) always have the same form as Eq.~\eqref{eq:qubit_jm_biased}. We intend to investigate this observation in future work.

The numerical findings of our strategy are summarised in Table~\ref{table_benchmarks} of the main text. To perform the numerical optimisation we use scipy \cite{scipy} implementation of the global optimisation strategy via differential evolution \cite{differential_evolution} and via Basin-hopping \cite{basinhopping} with Powell's method \cite{powell_method} and Sequential Least Squares Programming (SLSQP) \cite{numerical_optimization} as local optimisers.
We used qiskit \cite{qiskit} to handle the quantum chemistry Hamiltonians.

\section{Additional numerical details}\label{A:optimization_numbers}

In Table~\ref{appendix_table_benchmarks} below, we present the explicit quantities obtained from the numerical calculations (the counterpart of Table~\ref{table_benchmarks} in the main text, where the entries were normalised by the variance bound for the classical shadow strategy). 
Note that there are small (on the level of $10^{-5}$ relative error) discrepancies between classical shadows (CS) and unbiased joint measurability (JM) strategies.
In both cases, the calculation of the upper bound entails calculating operator norms of exponentially large matrices that are linear combinations of multiple Pauli terms (see Eq.~\eqref{appendix_eq:jm_norm}). 
In the CS strategy, all of the Pauli coefficients are calculated with very high precision (requiring simple multiplication of numbers), while in the JM strategy there is additional imprecision originating from the fact that the $\eta$ coefficients are solutions of SDPs. 
We therefore strongly suspect that both functions are analytically the same, and that the small errors in the coefficients obtained from the SDP aggregate to give the observed discrepancy in the operator norm.
In the main text we disregard these discrepancies.

\begin{table}[h!]
\begin{center}
\begin{tabular}{c|c|ccccc}
\textbf{Strategy/Molecule} & \textbf{En.} & & $\text{H}_2$ & LiH   & Be$\text{H}_2$ & $\text{H}_2$O \\ \hline
Classical shadows        & & &  68          & 144548  &  1480298           & 1999716    \\
Joint measurability                  & JW &  &  68          & 144549  &  1480313           & 1999735           \\
JM-optimised        & & & 68         & 8304  &   58294         & 198354     
\vspace{0.05cm}
\\ \hline
Classical shadows                 & & & 421          & 28446  & 103755            & 443442           \\
Joint measurability                  & BK & & 421          & 28446  & 103756            & 443445          \\
JM-optimised        & &  & 57          & 22274 & 57110           & 268506          
\vspace{0.05cm}
\\ \hline
Classical shadows                 &  & & 325          & 675792  &  2731597           & 3995396      \\
Joint measurability                  & P & & 325          & 675799  &  2731625           & 3995435          \\
JM-optimised        & & &123         & 14923  &   58294         & 67128           
\end{tabular}%
\end{center}
\caption{\label{appendix_table_benchmarks}
Explicit variance upper bounds obtained for various strategies. This extends Table~\ref{table_benchmarks} from the main text by including unbiased strategies for both classical shadows and joint measurements. Note that the entries are not normalised.
We interpret the relative differences between the unbiased CS and JM strategies, on the level of $10^{-5}$, as numerical artifacts. 
}
\end{table}

In all simulations we used measurement noise obtained via diagonal detector tomography \cite{maciejewski21} of single-qubit systems on IBM's Washington machine.

\end{document}